\documentclass[aps,preprint]{revtex4}%
\usepackage{amssymb}
\usepackage{amsfonts}
\usepackage{amsmath}
\usepackage{graphicx}%
\providecommand{\U}[1]{\protect\rule{.1in}{.1in}}
\newtheorem{theorem}{Theorem}

\newenvironment{proof}[1][Proof]{\noindent\textbf{#1.} }{\ \rule{0.5em}{0.5em}}
\begin{document}
\title{Quantum corrections to gravity}
\author{Yukio Tomozawa}
\date{\today }

\begin{abstract}
This paper revisits quantum corrections to gravity. It was shown previously by
other authors that quantum field theories in curved space time provide
quadratic curvature forms as quantum corrections to gravity in a conformally
flat metric. Application to a spherically symmetric and static (SSS) metric
shows that only the Gauss Bonnet combination (GB) yields the correct
expression. Using a variational method, the author shows that the metric he
obtained in 1985 as an example in a simplified case was indeed the exact
solution for a SSS metric. This proves that gravity becomes repulsive at short
distances by quantum corrections.

\end{abstract}

\maketitle

\section{Introduction}

Quantum field thories on a curved space time were extensively studied in
1970-80 in the case of a conformally flat (CF) metric, and it was concluded
that quantum corrections to gravity are expressed as quadratic curvature forms
in that metric\cite{bd}. The application of that formalism to the Friedman
Walker (FW) metric gave the interesting observation that a collapsed universe
will bouce back to an expansion in some cases\cite{anderson}. This provides a
hint that quantum corrections to gravity might yield a repulsive component at
short distances. The author studied quantum corrections for a SSS metric and
obtained a repulsive force at short distance in the special case of a Gauss
Bonnet (GB) combination, and also by a numerical computation. Based on this
analysis, the author proposed a model for high energy cosmic rays produced
from AGN in 1985\cite{cr1}, \cite{cr2}, \cite{cr3}. This turns out to predict
recent data from the Pierre Auger Project that suggests a possible correlation
between high energy cosmic rays and locations of AGN\cite{auger}. Further
discussion of the model yielded a new mass scale at the knee energy of the
cosmic ray spectrum and a prediction for the mass of a dark matter particle,
which is supported by recent data from HESS\cite{evidence}, \cite{agncr},
\cite{hess}. Because of successful predictions that agree with observational
reality, it is important to solidify a theoretical basis for the model. From
this point of view, the author has revisited the subject of quantum
corrections on gravity. Using an explicit form for a SSS metric, it is shown
that among proposed expressions in a conformally flat metric, only the Gauss
Bonnet combination satisfies the finiteness of renormalization. This implies
that the metric obtained by the author in 1985 is in fact the rigorous
solution for quantum corrections to gravity and the repulsive nature of
gravity at short distances has been established. A derivation of the final
result by a variational method is presented.

\section{Quantum corrections in a SSS metric}

Extensive studies of quantum corrections to gravity in a CF metric can be
summarized \cite{bd}, \cite{cr3} as%
\begin{equation}
H^{\mu\nu}=\frac{1}{(n-4)}\frac{1}{\sqrt{\left\vert g\right\vert }}%
\frac{\delta}{\delta g_{\mu\nu}}\int F\sqrt{\left\vert g\right\vert }d^{n}x,
\label{qcorr}%
\end{equation}
where%
\begin{align}
F  &  =AR^{\alpha\beta\gamma\delta}R_{\alpha\beta\gamma\delta}+BR^{\alpha
\beta}R_{\alpha\beta}+CR^{2}\\
&  =f_{0}+(n-4)f_{1}+O((n-4)^{2}), \label{qcorr2}%
\end{align}
and A, B and C are constants. It is well known that the GB combination%
\begin{equation}
A=1,B=-4,C=1 \label{gb}%
\end{equation}
is finite for any metric. For later reference, the GB combination is shown
explicitly as%
\begin{equation}
GB=R^{\alpha\beta\gamma\delta}R_{\alpha\beta\gamma\delta}-4R^{\alpha\beta
}R_{\alpha\beta}+R^{2}.
\end{equation}
For a CF metric, the combinations%
\begin{equation}
(A,B,C)=(0,1,0)
\end{equation}
and%
\begin{equation}
(A,B,C)=(0,0,1/3)
\end{equation}
give finite values and contribute to a quantum correction term, although the
two yield an identical correction in a CF metric. (It will be shown later that
the latter is an erroneous statement.) Using an explicit expression for
curvature in n dimensions, I will examine this property for a SSS metric in n
dimensions%
\begin{equation}
ds^{2}=e^{\nu}dt^{2}-e^{\lambda}dr^{2}-r^{2}d\Omega^{2}, \label{sss}%
\end{equation}
where $\nu$ and $\lambda$ are functions of the radial coordinate, $r$, and the
angular part can be expressed symbolically as%
\begin{equation}
d\Omega^{2}=d\omega_{1}^{2}+d\omega_{2}^{2}+\cdot\cdot\cdot\cdot+d\omega
_{n-2}^{2}.
\end{equation}

\begin{theorem}
For any metric, the $F$-term for quantum corrections is a quadratic curvature form.
\end{theorem}

\begin{proof}
For any metric other than a CF metric, the $F$-term can be written as%
\begin{equation}
F=F_{1}+\delta F, \label{result}%
\end{equation}
where $F_{1}$ is the $F$-term for quantum corrections to a CF metric, which
are quadratic curvature forms, as described above, and
\begin{equation}
\delta F=0
\end{equation}
for a CF metric. A CF metric is characterized by the vanishing of%
\begin{equation}
H=R^{\alpha\beta\gamma\delta}R_{\alpha\beta\gamma\delta}-2R^{\alpha\beta
}R_{\alpha\beta}+\frac{1}{3}R^{2}. \label{cf}%
\end{equation}
Then, $\delta F$ is nothing but%
\begin{equation}
\delta F=\lambda H,
\end{equation}
where $\lambda$ is a non-dimensional constant, since this is the only
dimensionally accessible quantity. In other words, Eq. (\ref{result}) is a sum
of quadratic curvature forms for any metric.
\end{proof}

In the rest of this section, it will be shown that the GB combination is the
only possible solution if Eq. (\ref{qcorr}) is to be finite in the limit as n
approaches 4 for a SSS metric. For that purpose, I will start with Eq.
(\ref{qcorr}) and Eq. (\ref{qcorr2}) for a SSS metric.

\begin{theorem}
The $F$-term for quantum correction of a SSS metric is a GB quadratic combination.
\end{theorem}

\begin{proof}
From the computation of curvatures in the appendix, one can express the
integrand in Eq. (\ref{qcorr2}), where $f_{0}$ and $f_{1\text{ }}$are
coefficients in the expansion in (n-4) and%
\begin{align}
f_{0}  &  =(A+B/2+C)g_{1}+(B+4C)g_{2}\nonumber\\
&  -4Cg_{3},
\end{align}
where%
\begin{align}
g_{1}  &  =(\nu^{\prime\prime}+\nu^{\prime2}/2-\nu^{\prime}\lambda^{\prime
}/2)^{2}e^{-2\lambda}+2(\nu^{\prime2}+\lambda^{\prime2})(e^{-\lambda}%
/r)^{2}\nonumber\\
&  +4((e^{\lambda}-1)e^{-\lambda}/r^{2})^{2},
\end{align}%
\begin{align}
g_{2}  &  =((\nu^{\prime\prime}+\nu^{\prime2}/2-\nu^{\prime}\lambda^{\prime
}/2)(\nu^{\prime}-\lambda^{\prime})-\nu^{\prime}\lambda^{\prime}%
/r)e^{-2\lambda}/r\nonumber\\
&  +\frac{1}{2}(\nu^{\prime2}+\lambda^{\prime2})(e^{-\lambda}/r)^{2}%
-2(\nu^{\prime}-\lambda^{\prime})((e^{\lambda}-1)/r^{2})e^{-2\lambda}/r,.
\end{align}
and%
\begin{equation}
g_{3}=((\nu^{\prime\prime}+\nu^{\prime2}/2-\nu^{\prime}\lambda^{\prime
}/2)(e^{\lambda}-1)+\nu^{\prime}\lambda^{\prime})e^{-2\lambda}/r^{2}.
\end{equation}
First manipulate the last term. Multiplying by the r dependent term,
\begin{equation}
\sqrt{\left\vert g\right\vert }=e^{(\nu+\lambda)/2}r^{n-2}%
\end{equation}
and noticing that a total derivative does not contribute to a variation, one
can reformulate%
\begin{align}
G  &  =(\nu^{\prime\prime}+\nu^{\prime2}/2)(e^{\lambda}-1)e^{-2\lambda}%
e^{(\nu+\lambda)/2}r^{n-4}\nonumber\\
&  =2(e^{\nu/2})^{\prime\prime}(e^{-\lambda/2}-e^{-3\lambda/2})r^{n-4}%
\nonumber\\
&  =-(2(e^{\nu/2})^{\prime}((e^{-\lambda/2}-e^{-3\lambda/2})r^{n-4})^{\prime
}\nonumber\\
&  =\nu^{\prime}\lambda^{\prime}(e^{\lambda}-3)/2e^{-2\lambda}e^{(\nu
+\lambda)/2}r^{n-4}\nonumber\\
&  -(n-4)\nu^{\prime}(e^{\lambda}-1)e^{-2\lambda}e^{(\nu+\lambda)/2}r^{n-5}%
\end{align}
then%
\begin{align}
&  -4Cg_{3}e^{(\nu+\lambda)/2}r^{n-2}\nonumber\\
&  =4C(n-4)\nu^{\prime}(e^{\lambda}-1)e^{-2\lambda}e^{(\nu+\lambda)/2}r^{n-5},
\label{n-4}%
\end{align}
where all total derivative terms are dropped.

Let us show that the last term of $g_{1}$, which does not contain a
derivative, cannot be eliminated by manipulation of a partial derivative of
$g_{1}e^{(\nu+\lambda)/2}r^{n-2}$. This would be possible if differentiation
of the $r^{m}$ term were to be performed and the derivative order were to be
reduced. However, by doing that a factor of $n-4$ would be multiplied in the
transition from the second derivative to the first derivative (or from a
product of two first order derivatives to a first order derivative).
Therefore, the last term of $g_{1}$ cannot be eliminated by partial
differentiation, since all the terms without derivatives which arise in such a
manner contain n-4 as a factor, and hence vanish for n approaching 4. By the
same argument, $g_{1}$ does not contain terms with a first order derivative.
Likewise, the last term of $g_{2}$ cannot be eliminated by the technique of
partial differentiation.

The last term of $g_{2}e^{(\nu+\lambda)/2}r^{n-2}$ can be split into two
terms. While one of them is reduced to%
\begin{equation}
(\nu^{\prime}-\lambda^{\prime})e^{(\nu-\lambda)/2}r^{n-5}=-2(n-5)e^{(\nu
-\lambda)/2}r^{n-6}%
\end{equation}
by partial differentiation, the other term is expressed as%
\begin{align}
(\nu^{\prime}-\lambda^{\prime})e^{(\nu-3\lambda)/2}r^{n-5}  &  =2\lambda
^{\prime}e^{(\nu-3\lambda)/2}r^{n-5}-2(n-5)e^{(\nu-3\lambda)/2}r^{n-6}%
\label{derivative1}\\
&  =\frac{2}{3}\nu^{\prime}e^{(\nu-3\lambda)/2}r^{n-5}-\frac{2}{3}%
(n-5)e^{(\nu-3\lambda)/2}r^{n-6}. \label{derivative2}%
\end{align}
From a linear combination of the two terms on the right hand side, one gets%
\begin{align}
(\nu^{\prime}-\lambda^{\prime})e^{(\nu-3\lambda)/2}r^{n-5}  &  =\frac
{2}{3(a+b)}((a\nu^{\prime}+3b\lambda^{\prime})e^{(\nu-3\lambda)/2}%
r^{n-5}\nonumber\\
&  -(a+3b)(n-5)e^{(\nu-3\lambda)/2}r^{n-6})\\
&  =(\nu^{\prime}-\lambda^{\prime})e^{(\nu-3\lambda)/2}r^{n-5}\nonumber\\
&  -\frac{a+3b}{3(a+b)}((\nu^{\prime}-3\lambda^{\prime})e^{(\nu-3\lambda
)/2}r^{n-5}+2(n-5)e^{(\nu-3\lambda)/2}r^{n-6})\\
&  =(\nu^{\prime}-\lambda^{\prime})e^{(\nu-3\lambda)/2}r^{n-5}.
\end{align}
This means that the first derivative term in Eqs. (\ref{derivative1}) or
(\ref{derivative2}) cannot be eliminated by partial differentiation. Since the
last term of $g_{2}$ is the only term which contains a first derivative, and
cannot be eliminated, the coefficient of $g_{2}$ must vanish.%
\begin{equation}
B+4C=0
\end{equation}
Then the last term of $g_{1}$ is the only term which contains a non-derivative
term and cannot be eliminated by partial differentiation, so the coefficient
of $g_{1}$ should vanish%
\begin{equation}
A+B/2+C=0
\end{equation}
These results yield%
\begin{equation}
B=-4C,\text{ \ \ and \ \ }A=C
\end{equation}
which is a multiple of the GB condition, Eq. (\ref{gb}). The last term of
$f_{0}$ is proportional to $(n-4)$ as is seen in Eq. (\ref{n-4}). This term
should be added to the quantum correction term $f_{1}$ in the variational
calculation. This proves that quantum corrections for a SSS metric should
consist of a GB combination exclusively, where%
\begin{equation}
f_{0}=0.
\end{equation}
This completes the proof of the Theorem.
\end{proof}

\ \ \ \ \ \ \ \ \ \ \ \ \ \ \ \ \ \ \ \ \ \ \ \ \ \ \ \ \ \ \ \ \ \ \ \ \ \ \ \ \ \ \ \ \ \ \ \ \ \ \ \ \ \ \ \ \ \ \ \ \ \ \ \ \ \ \ \ \ \ \ \ \ \ \ \ \ \ \ \ \ \ \ \ \ \ \ \ \ \ \ \ \ \ \ \ \ \ \ \ \ \ \ \ \ \ \ \ \ \ \ \ \ \ \ \ \ \ \ \ \ \ \ \ \ \ \ \ \ \ \ \ \ \ \ \ \ \ \ \ \ \ \ \ \ \ \ \ \ \ \ \ \ \ \ \ \ \ \ \ \ \ \ \ \ \ \ \ \ \ \ \ \ \ \ \ \ \ \ \ \ \ \ \ \ \ \ \ \ \ \ \ \ \ \ \ \ \ 

\ \ Before proceeding to the computation of $f_{1}$ for the purpose of the
variational calculation, a comment on the relationship between the two
quadratic curvature terms is in order. From the relationship%
\begin{equation}
R^{\alpha\beta}R_{\alpha\beta}-\frac{1}{3}R^{2}=\frac{1}{2}(H-GB),
\end{equation}
the difference between $R^{\alpha\beta}R_{\alpha\beta}$ and $\frac{1}{3}R^{2}$
in a CF metric is $-\frac{1}{2}GB$, instead of zero, as is claimed in a
literature\cite{bd}.

Using expressions from the appendix, the expansion of $F$ in $(n-4)$ yields%
\begin{align}
f_{1}  &  =A((\nu^{\prime2}+\lambda^{\prime2})(e^{-\lambda}/r)^{2}%
+6(e^{\lambda}-1)^{2}(e^{-\lambda}/r)^{2})\nonumber\\
&  +B(\frac{1}{2}(\nu^{\prime\prime}+\nu^{\prime2}/2-\nu^{\prime}%
\lambda^{\prime}/2)(\nu^{\prime}-\lambda^{\prime})e^{-2\lambda}/r\nonumber\\
&  +\frac{5}{4}(\nu^{\prime2}+\lambda^{\prime2})(e^{-\lambda}/r)^{2}%
-\nu^{\prime}\lambda^{\prime}(e^{-\lambda}/r)^{2}/2\nonumber\\
&  -3(\nu^{\prime}-\lambda^{\prime})((e^{\lambda}-1)/r)(e^{-\lambda}%
/r)^{2}\nonumber\\
&  +5((e^{\lambda}-1)/r)^{2}(e^{-\lambda}/r)^{2})\nonumber\\
&  +C(2(\nu^{\prime\prime}+\nu^{\prime2}/2-\nu^{\prime}\lambda^{\prime}%
/2)(\nu^{\prime}-\lambda^{\prime})e^{-2\lambda}/r\nonumber\\
&  +4(\nu^{\prime2}+\lambda^{\prime2})(e^{-\lambda}/r)^{2}-8\nu^{\prime
}\lambda^{\prime}(e^{-\lambda}/r)^{2}\nonumber\\
&  -6(\nu^{\prime\prime}+\nu^{\prime2}/2-\nu^{\prime}\lambda^{\prime
}/2)(e^{\lambda}-1)(e^{-\lambda}/r)^{2}\nonumber\\
&  -16(\nu^{\prime}-\lambda^{\prime})((e^{\lambda}-1)/r)(e^{-\lambda}%
/r)^{2}\nonumber\\
&  +12(e^{\lambda}-1)^{2}(e^{-\lambda}/r^{2})^{2})\nonumber\\
&  +4C\nu^{\prime}((e^{\lambda}-1)/r)(e^{-\lambda}/r)^{2},
\end{align}
where the last term is added from Eq. (\ref{n-4}). Rearranging appropriate
terms, one gets%
\begin{align}
f_{1}  &  =(\nu^{\prime2}+\lambda^{\prime2})(e^{-\lambda}/r)^{2}%
(A+B/4+B+4C)\nonumber\\
&  +(B+4C)g_{4}\nonumber\\
&  -6Cg_{3}\nonumber\\
&  +((e^{\lambda}-1)(e^{-\lambda}/r^{2}))^{2}(6A+5B+12C)\nonumber\\
&  -(\nu^{\prime}-\lambda^{\prime})((e^{\lambda}-1)/r)(e^{-\lambda}%
/r)^{2}(3B+16C)\nonumber\\
&  +4C\nu^{\prime}((e^{\lambda}-1)/r)(e^{-\lambda}/r)^{2},
\end{align}
where%
\begin{equation}
g_{4}=\frac{1}{2}((\nu^{\prime\prime}+\nu^{\prime2}/2-\nu^{\prime}%
\lambda^{\prime}/2)(\nu^{\prime}-\lambda^{\prime})r-\nu^{\prime}%
\lambda^{\prime})(e^{-\lambda}/r)^{2}%
\end{equation}
In this expression, the first two terms vanish by the BG condition, Eq.
(\ref{gb}), while the third term becomes higher order in $(n-4)$ by Eq.
(\ref{n-4}). In the last two terms, the $\nu^{\prime}$ terms cancel and
finally one gets%
\begin{equation}
f_{1}=4\lambda^{\prime}((e^{\lambda}-1)/r)(e^{-\lambda}/r)^{2}-2(e^{\lambda
}-1)^{2}(e^{-\lambda}/r^{2})^{2}.
\end{equation}

Multiplying by $e^{(\nu+\lambda)/2}r^{2}$, variation of%
\begin{equation}
g_{00}=e^{\nu}%
\end{equation}
yields%
\begin{align}
H^{00}  &  =\frac{1}{e^{(\nu+\lambda)/2}r^{2}}\frac{\delta}{\delta(e^{\nu}%
)}(f_{1}e^{(\nu+\lambda)/2}r^{2})\\
&  =\frac{1}{2}e^{-\nu}f_{1}%
\end{align}
or%
\begin{align}
H_{0}^{0}  &  =\frac{1}{2}f_{1}\label{rhs1}\\
&  =2\lambda^{\prime}((e^{\lambda}-1)/r)(e^{-\lambda}/r)^{2}-(e^{\lambda
}-1)^{2}(e^{-\lambda}/r^{2})^{2}.
\end{align}

For the variation of%
\begin{equation}
g_{11}=-e^{\lambda},
\end{equation}

\begin{align}
H^{11}  &  =\frac{1}{e^{(\nu+\lambda)/2}r^{2}}(\frac{\delta}{\delta
(-e^{\lambda})}(f_{1}e^{(\nu+\lambda)/2}r^{2})\nonumber\\
&  -(\frac{\delta}{\delta(-(e^{\lambda})^{\prime})}(f_{1}e^{(\nu+\lambda
)/2}r^{2}))^{\prime})\nonumber\\
&  =\frac{1}{e^{(\nu+\lambda)/2}r^{2}}(4(e^{\lambda})^{\prime}\frac{\frac
{3}{2}e^{-5\lambda/2}-\frac{5}{2}e^{-7\lambda/2}}{r}e^{\nu/2}\nonumber\\
&  +(e^{-\lambda/2}+2e^{-3\lambda/2}-3e^{-5\lambda/2})e^{\nu/2}\nonumber\\
&  +4(\frac{e^{-3\lambda/2}-e^{-5\lambda/2}}{r})e^{\nu/2})^{\prime
})\nonumber\\
&  =\frac{1}{e^{(\nu+\lambda)/2}r^{2}}(2\nu^{\prime}\frac{e^{\lambda}-1}%
{r}e^{-3\lambda}e^{(\nu+\lambda)/2}\nonumber\\
&  +(e^{\lambda}-1)^{2}e^{-3\lambda}e^{(\nu+\lambda)/2}),
\end{align}
and then%
\begin{align}
H_{1}^{1}  &  =-e^{\lambda}H^{11}\label{rhs2}\\
&  =-2\nu^{\prime}((e^{\lambda}-1)/r)(e^{-\lambda}/r)^{2}-(e^{\lambda}%
-1)^{2}(e^{-\lambda}/r^{2})^{2}.
\end{align}
These expressions, Eq. (\ref{rhs1}) and Eq. (\ref{rhs2}), have been derived by
a different method in 1985\cite{cr1},\cite{cr3}, and were used for solving the
Einstein equation. This process will be recapitulated in the next section. The
difference from 1985 is that this time we have shown that this is the exact
solution for a SSS metric.

Finally one can estimate the coefficient of the GB term for a SSS metric. In a
CF metric, the $F$-term is expressed as%
\begin{equation}
\frac{\alpha}{6}R^{2}-\beta GB,
\end{equation}
where\cite{bd}%
\begin{equation}
\alpha=\frac{1}{2880\pi^{2}}(N_{S}+6N_{\nu}-18N_{V}),
\end{equation}
and%
\begin{equation}
\beta=\frac{1}{2880\pi^{2}}(N_{S}+11N_{\nu}+62N_{V}),
\end{equation}
$N_{s}$, $N_{\nu}$ and $N_{V}$ being the numbers of scalar fields,
four-component neutrino fields and vector fields respectively. Then the
$F$-term in a SSS metric becomes%
\begin{align}
\frac{\alpha}{6}R^{2}-\beta GB  &  \rightarrow\frac{\alpha}{6}(R^{2}%
-3R^{\alpha\beta}R_{\alpha\beta}+\frac{3}{2}H)-\beta GB\\
&  =-\kappa GB,
\end{align}
where%
\begin{align}
\kappa &  =\beta-\frac{1}{4}\alpha\\
&  =\frac{1}{2880\pi^{2}}(\frac{3}{4}N_{s}+\frac{19}{2}N_{\nu}+\frac{133}%
{2}N_{V}). \label{kappa}%
\end{align}

Eq. (49) indicates that the quantum corrections for any metric are the GB
combination of quadratic curvature forms, unless it is a CF metric.

\section{Explicit Form of Quantum Corrections for a SSS Metric}

The Einstein equation for a SSS metric, Eq. (\ref{sss}), is expressed%
\begin{equation}
G_{\rho}^{\mu}=R_{\rho}^{\mu}-\delta_{\rho}^{\mu}R=-\frac{l^{2}\kappa}%
{2}H_{\rho}^{\mu},
\end{equation}
\bigskip where\cite{bd}%
\begin{equation}
l^{2}=16\pi G
\end{equation}
and $\kappa$ is given in Eq. (\ref{kappa}). An explicit form of the Einstein
equation\bigskip\ reads%
\begin{align}
&  e^{-\lambda}(\frac{\lambda^{\prime}}{r}-\frac{1}{r^{2}})+\frac{1}{r^{2}%
}\nonumber\\
&  =\frac{\xi}{2}(e^{-2\lambda}(\frac{2\lambda^{\prime}}{r^{3}}+\frac{1}%
{r^{4}})+2e^{-\lambda}(-\frac{\lambda^{\prime}}{r^{3}}-\frac{1}{r^{4}}%
)+\frac{1}{r^{4}}) \label{ee1}%
\end{align}
and%
\begin{align}
&  e^{-\lambda}(\frac{-\nu^{\prime}}{r}-\frac{1}{r^{2}})+\frac{1}{r^{2}%
}\nonumber\\
&  =\frac{\xi}{2}(e^{-2\lambda}(\frac{-2\nu^{\prime}}{r^{3}}+\frac{1}{r^{4}%
})+2e^{-\lambda}(\frac{\nu^{\prime}}{r^{3}}-\frac{1}{r^{4}})+\frac{1}{r^{4}}),
\end{align}
where%
\begin{equation}
\xi=l^{2}\kappa.
\end{equation}
Subtraction of the two Einstein equations yields%
\begin{equation}
\nu^{\prime}+\lambda^{\prime}=0
\end{equation}
\bigskip or%
\begin{equation}
\nu+\lambda=0
\end{equation}
by the boundary condition at%
\begin{equation}
r=\infty.
\end{equation}
Multiplying Eq. (\ref{ee1}) by $r^{2}$,%
\begin{align}
e^{-\lambda}(r\lambda^{\prime}-1)+1  &  =\frac{\xi}{2}(e^{-2\lambda}%
(\frac{2\lambda^{\prime}}{r}+\frac{1}{r^{2}})\nonumber\\
&  +2e^{-\lambda}(-\frac{\lambda^{\prime}}{r}-\frac{1}{r^{2}})+\frac{1}{r^{2}%
}),
\end{align}
which can be expressed as%
\begin{align}
-(re^{-\lambda})^{\prime}+r^{\prime}  &  =\frac{\xi}{2}(-(\frac{e^{-2\lambda}%
}{r})^{\prime}\nonumber\\
&  +2(\frac{e^{-\lambda}}{r})^{\prime}-(\frac{1}{r})^{\prime}),
\end{align}
which can be integrated%
\begin{align}
&  -re^{-\lambda}+r\nonumber\\
&  =\frac{\xi}{2}(-\frac{e^{-2\lambda}}{r}+2\frac{e^{-\lambda}}{r}-\frac{1}%
{r})+K,
\end{align}
or%
\begin{equation}
e^{-2\lambda}-2(1+\frac{r^{2}}{\xi})e^{-\lambda}+1+\frac{2r^{2}}{\xi}%
-\frac{2Kr}{\xi}=0,
\end{equation}
where K stands for an integration constant. The solution of this quadratic
equation is%
\begin{equation}
e^{-\lambda}=e^{\nu}=1+\frac{r^{2}}{\xi}-\sqrt{\frac{r^{4}}{\xi^{2}}%
+\frac{2Kr}{\xi}}. \label{solution}%
\end{equation}
This is the solution\cite{cr1},\cite{cr3} that was obtained by the author in
1985. In the limit, $r\rightarrow\infty$ or $\xi\rightarrow0$, one obtains%
\begin{equation}
e^{\nu}\rightarrow1-\frac{K}{r}+\frac{K^{2}\xi}{2r^{4}}+\cdot\cdot\cdot,
\end{equation}
so that the integration constant K is determined to be%
\begin{equation}
K=2GM.
\end{equation}
The important point is that this is the exact solution for quantum corrections
to a SSS metric. This solution was obtained by the author in August,
1985\cite{cr1}. The same solution was obtained independently by Boulware and
Deser\cite{bou-d} for a string-generated gravity model.

\section{Repulsive Nature of Gravity at Short Distances}

As $r\rightarrow0$, the solution of Eq. (\ref{solution}) becomes%
\begin{equation}
e^{\nu}\rightarrow1-\sqrt{\frac{2Kr}{\xi}}%
\end{equation}
which shows the repulsive nature of quantum corrections to gravity at short
distances. While the outer horizon is at K=2GM, the inner horizen is%
\begin{equation}
\frac{\xi}{2K}.
\end{equation}

If the sign of $\xi$ is reversed, the solution becomes%
\begin{equation}
e^{-\lambda}=e^{\nu}=1-\frac{r^{2}}{\mid\xi\mid}+\sqrt{\frac{r^{4}}{\xi^{2}%
}-\frac{2Kr}{\mid\xi\mid}}%
\end{equation}
The solution is terminated at%
\begin{equation}
r=(2K\shortmid\xi\shortmid)^{1/3},\label{inside}%
\end{equation}
at which point%
\begin{equation}
e^{-\lambda}=e^{\nu}=-(\frac{2K}{\sqrt{\shortmid\xi\shortmid}})^{2/3}.
\end{equation}
In other words, the gravitational potential is attractive up to the radius $r$
in Eq. (\ref{inside}), inside of which there is no solution. This is
equivalent to a repusive core at that radius. Hence, irrespective of the sign
of $\xi$, the quantum corrections make gravity repulsive at short distances.
This is a useful information, since the sign convention in general relativity
is reversed from author to author, as to that of the curvature and/or the
right hand side of the Einstein equation.

As is well known by now, this result is the basis for the author's model of
high energy cosmic ray emission from AGN\cite{cr1}-\cite{cr3}.

\section{Appendix: Curvature in a SSS Metric in n Dimensions}

The extension of the curvature in a SSS metric to n dimensions reads%
\begin{equation}
R_{....01}^{01}=(\frac{\nu^{\prime\prime}}{2}+\frac{\nu^{\prime2}}{4}%
-\frac{\nu^{\prime}\lambda^{\prime}}{4})e^{-\lambda}%
\end{equation}

\begin{align}
R_{....02}^{02}  &  =R_{....03}^{03}=\cdot\cdot\cdot\nonumber\\
&  =R_{....0,n-1}^{0,n-1}=\frac{\nu^{\prime}}{2}\frac{e^{-\lambda}}{r}%
\end{align}

\begin{align}
R_{....12}^{12}  &  =R_{....13}^{13}=\cdot\cdot\cdot\nonumber\\
&  =R_{....1,n-1}^{1,n-1}=-\frac{\lambda^{\prime}}{2}\frac{e^{-\lambda}}{r}%
\end{align}

\begin{align}
R_{....23}^{23}  &  =R_{....24}^{24}=\cdot\cdot\cdot=R_{....2,n-1}%
^{2,n-1}=R_{....34}^{34}\nonumber\\
&  =\cdot\cdot\cdot=R_{....n-2,n-1}^{n-2,n-1}=-(e^{\lambda}-1)\frac
{e^{-\lambda}}{r^{2}}%
\end{align}

\begin{align}
R_{..0}^{0}  &  =R_{....01}^{01}+R_{....02}^{02}+\cdot\cdot\cdot
+R_{....0,n-1}^{0,n-1}\nonumber\\
&  =(\frac{\nu^{\prime\prime}}{2}+\frac{\nu^{\prime2}}{4}-\frac{\nu^{\prime
}\lambda^{\prime}}{4})e^{-\lambda}+(n-2)\frac{\nu^{\prime}}{2}\frac
{e^{-\lambda}}{r}%
\end{align}

\begin{align}
R_{..1}^{1}  &  =R_{....10}^{10}+R_{....12}^{12}+R_{....13}^{13}+\cdot
\cdot\cdot+R_{....1,n-1}^{1,n-1}\nonumber\\
&  =(\frac{\nu^{\prime\prime}}{2}+\frac{\nu^{\prime2}}{4}-\frac{\nu^{\prime
}\lambda^{\prime}}{4})e^{-\lambda}-(n-2)\frac{\lambda^{\prime}}{2}%
\frac{e^{-\lambda}}{r}%
\end{align}

\begin{align}
R_{..2}^{2}  &  =R_{....20}^{20}+R_{....21}^{21}+R_{....23}^{23}+\cdot
\cdot\cdot+R_{....2,n-1}^{2,n-1}\nonumber\\
&  =\frac{\nu^{\prime}-\lambda^{\prime}}{2}\frac{e^{-\lambda}}{r}%
-(n-3)(e^{\lambda}-1)\frac{e^{-\lambda}}{r^{2}}\nonumber\\
&  =R_{..3}^{3}=R_{..4}^{4}=\cdot\cdot\cdot=R_{..n-1}^{n-1}%
\end{align}

\begin{align}
R  &  =R_{..0}^{0}+R_{..1}^{1}+\cdot\cdot\cdot+R_{..n-1}^{n-1}\nonumber\\
&  =(\nu^{\prime\prime}+\frac{\nu^{\prime2}}{2}-\frac{\nu^{\prime}%
\lambda^{\prime}}{2})e^{-\lambda}+(n-2)(\nu^{\prime}-\lambda^{\prime}%
)\frac{e^{-\lambda}}{r}\nonumber\\
&  -(n-2)(n-3)(e^{\lambda}-1)\frac{e^{-\lambda}}{r^{2}}%
\end{align}
For the quadratic curvature form,%
\begin{align}
&  R_{....\gamma\delta}^{\alpha\beta}R_{\alpha\beta}^{....\gamma\delta
}\nonumber\\
&  =4((R_{....01}^{01})^{2}+(R_{....02}^{02})^{2}+\cdot\cdot\cdot
+(R_{....0,n-1}^{0,n-1})^{2}\nonumber\\
&  +(R_{....12}^{12})^{2}+(R_{....13}^{13})^{2}+\cdot\cdot\cdot+(R_{....1,n-1}%
^{1,n-1})^{2}\nonumber\\
&  +(R_{....23}^{23})^{2}+(R_{....24}^{24})^{2}+\cdot\cdot\cdot+(R_{....2,n-1}%
^{2,n-1})^{2}\nonumber\\
&  +(R_{....34}^{34})^{2}+\cdot\cdot\cdot\nonumber\\
&  +\cdot\cdot\cdot\nonumber\\
&  +(R_{....n-2,n-1}^{n-2,n-1})^{2})\nonumber\\
&  =((\nu^{\prime\prime}+\frac{\nu^{\prime2}}{2}-\frac{\nu^{\prime}%
\lambda^{\prime}}{2})e^{-\lambda})^{2}+(n-2)(\nu^{\prime2}+\lambda^{\prime
2})(\frac{e^{-\lambda}}{r})^{2}\nonumber\\
&  +2(n-2)(n-3)((e^{\lambda}-1)\frac{e^{-\lambda}}{r^{2}})^{2}%
\end{align}

\begin{align}
R_{..\beta}^{\alpha}R_{\alpha}^{..\beta}  &  =(R_{..0}^{0})^{2}+(R_{..1}%
^{1})^{2}+\cdot\cdot\cdot(R_{..n-1}^{n-1})^{2}\nonumber\\
&  =\frac{1}{2}((\nu^{\prime\prime}+\frac{\nu^{\prime2}}{2}-\frac{\nu^{\prime
}\lambda^{\prime}}{2})e^{-\lambda})^{2}\nonumber\\
&  +\frac{n-2}{2}(\nu^{\prime\prime}+\frac{\nu^{\prime2}}{2}-\frac{\nu
^{\prime}\lambda^{\prime}}{2})(\nu^{\prime}-\lambda^{\prime})\frac
{e^{-2\lambda}}{r}\nonumber\\
&  +\frac{(n-2)^{2}}{4}(\nu^{\prime2}+\lambda^{\prime2})(\frac{e^{-\lambda}%
}{r})^{2}\nonumber\\
&  +\frac{(n-2)}{4}(\nu^{\prime}-\lambda^{\prime})^{2}(\frac{e^{-\lambda}}%
{r})^{2}\nonumber\\
&  -(n-2)(n-3)(\nu^{\prime}-\lambda^{\prime})\frac{e^{\lambda}-1}{r}%
(\frac{e^{-\lambda}}{r})^{2}\nonumber\\
&  +(n-2)(n-3)^{2}(\frac{e^{\lambda}-1}{r})^{2}(\frac{e^{-\lambda}}{r})^{2}%
\end{align}
and%
\begin{align}
R^{2}  &  =((\nu^{\prime\prime}+\frac{\nu^{\prime2}}{2}-\frac{\nu^{\prime
}\lambda^{\prime}}{2})e^{-\lambda}+(n-2)(\nu^{\prime}-\lambda^{\prime}%
)\frac{e^{-\lambda}}{r}\nonumber\\
&  -(n-2)(n-3)(e^{\lambda}-1)\frac{e^{-\lambda}}{r^{2}})^{2}.
\end{align}

\begin{acknowledgments}
It is a great pleasure to thank David N. Williams for reading the
manuscript.\bigskip
\end{acknowledgments}

\end{document}